\documentclass[11pt]{article}

\usepackage{amsmath, amsfonts, amssymb}

\usepackage{amsthm}

\usepackage{fullpage}
\usepackage{mdframed}

\usepackage{graphicx}

\usepackage{comment}

\usepackage{tikz}
\usetikzlibrary{positioning}
\usetikzlibrary{fit}

\newtheorem{theorem}{Theorem}[section]
\newtheorem{lemma}[theorem]{Lemma}
\newtheorem{corollary}[theorem]{Corollary}

\newtheorem{defn}{Definition}[section]

\newcommand{\A}{\mathcal{A}}
\newcommand{\B}{\mathcal{B}}
\newcommand{\C}{\mathcal{C}}
\newcommand{\G}{\mathcal{G}}
\newcommand{\E}{\mathcal{E}}
\newcommand{\M}{\mathcal{M}}
\renewcommand{\P}{\mathcal{P}}
\newcommand{\Q}{\mathcal{Q}}
\newcommand{\T}{\mathcal{T}}
\newcommand{\U}{\mathcal{U}}
\newcommand{\V}{\mathcal{V}}
\newcommand{\X}{\mathcal{X}}
\newcommand{\Z}{\mathcal{Z}}

\renewcommand{\Lambda}{M}

\newcommand{\EE}{\mathbf{E}}

\newcommand{\NN}{\mathbb{N}}
\newcommand{\RR}{\mathbb{R}}

\newcommand{\Lap}{\textrm{Lap}}
\newcommand{\SP}{\textrm{SP}}
\newcommand{\polylog}{\textrm{polylog}}

\newcommand{\hd}{h} 
\newcommand{\hl}{\ell} 
\newcommand{\binset}{\{0,1\}} 

\newcommand{\wone} {\|w\|_1}

\newcounter{algorithm}

\newenvironment{algorithm}[1][]{
\refstepcounter{algorithm}
$$
$$
\begin{mdframed}
\begin{center}
{\bf Algorithm \thealgorithm: #1}
\end{center}
\newcommand{\Inputs}[1]{{\bf Inputs}: ##1}

}
{
\end{mdframed}
}

\title{Shortest Paths and Distances with Differential Privacy}
\author{
Adam Sealfon \\
{\normalsize MIT} \\
{\normalsize asealfon@csail.mit.edu}
}

\date{}

\begin{document}
\pagenumbering{gobble}
\clearpage
\maketitle

\begin{abstract}
We introduce a model for differentially private analysis of weighted graphs in which the graph topology $(\V,\E)$ is
assumed to be public and the private information consists only of the edge weights $w:\E\to\mathbb{R}^+$.
This can express hiding congestion patterns in a known system of roads.
Differential privacy requires that the output of an algorithm provides little advantage,
measured by privacy parameters $\epsilon$ and $\delta$, for distinguishing between neighboring
inputs, which are thought of as inputs that differ on the contribution of one individual.
In our model,
two weight functions $w,w'$ are considered to be neighboring if 
they have $\ell_1$ distance at most one.

We study the problems of privately releasing a short path between a pair of vertices
and of privately releasing approximate distances between all pairs of vertices.
We are concerned with the {\em approximation error}, the difference between
the length of the released path or released distance and
the length of the shortest path or actual distance.

For the problem of privately releasing a short path between 
a pair of vertices,
we prove a lower bound of $\Omega(|\V|)$ on the additive approximation error 
for fixed privacy parameters $\epsilon,\delta$.
We provide a differentially private algorithm that matches this error bound up to a logarithmic factor
and releases paths between all pairs of vertices, not just a single pair.
The approximation error achieved by our algorithm can be bounded by the
number of edges on the shortest path, so we achieve better accuracy than the worst-case bound
for pairs of vertices that are connected by a low-weight path consisting of $o(|\V|)$ vertices.

For the problem of privately releasing all-pairs distances, 
we show that for trees we can release all-pairs distances with approximation error $O(\log^{2.5}|\V|)$
for fixed privacy parameters.
For arbitrary bounded-weight graphs with edge weights in $[0,M]$ we can
release all distances with approximation error $\tilde{O}(\sqrt{|\V|M})$.
\end{abstract}

\newpage
\tableofcontents
\newpage
\pagenumbering{arabic}

\vspace{0.3in}
\section{Introduction}

\subsection{Differential privacy and our model}
Privacy-preserving data analysis is concerned with releasing useful aggregate information from a 
database while protecting sensitive information about individuals in the database.
\emph{Differential privacy} \cite{DMNS06} provides strong privacy guarantees while allowing
many types of queries to be answered approximately.
Differential privacy requires that for any pair of \textit{neighboring} databases $x$ and $y$, 
the output distributions of the
mechanism on database $x$ and database $y$ are very close. 
In the traditional setting, 
databases are collections of data records and are considered to be
neighboring if they are identical except for a single record, which is thought of as
the information associated with a single individual.
This definition guarantees that
an adversary can learn very little about any individual in the database,
no matter what auxiliary information the adversary may possess.

\paragraph{A new model}
We introduce a model of differential privacy for the setting in which databases are weighted graphs $(\G,w)$.
In our setting, the graph topology $\G=(\V,\E)$ is assumed to be public, and only the edge weights 
$w:\E\to\RR^+$ must be kept private. Individuals are assumed to have only bounded influence on the edge weights.
Consequently, two weight functions 
on the same unweighted graph 
are considered to be neighbors if they have $\ell_1$ distance at most one.

The model is well suited to capture the setting of a navigation system
which has access to current traffic data and uses it to direct
drivers. The travel times along routes provided by the system should be as short as possible.  
However, the traffic data used by the system may consist of private information. 
For instance, navigation tools such as Google Maps and Waze estimate traffic conditions based on GPS locations
of cars or based on traffic reports by users. 
On the other hand, the topology of the road map used by the system 
is non-private, since it is a static graph readily available to all users.
We would like to provide information about paths and distances in the network without revealing
sensitive information about the edge weights.
In this paper we introduce a variant of differential privacy suitable for studying network routing,
and provide both algorithms and lower bounds.

\subsection{Our results}
We will consider two classic problems in this model which are relevant for routing.
First, we are interested in finding \emph{short paths} between specified pairs of vertices. 
We cannot hope always to release the shortest path while preserving privacy, 
but we would like to release a path that is not much longer than the shortest path.
Here the \emph{approximation error}\hspace{-0.05px} is the difference in length between the released path and the shortest path.
Second, we would like to release \emph{distances} between
pairs of vertices.  The \emph{approximation error} here is the absolute difference between the released distance
and the actual distance. As we discuss below, releasing an accurate distance estimate for a single pair of
vertices in our model is a straightforward application of the Laplace mechanism of \cite{DMNS06}. 
We focus on the more difficult problem of releasing \emph{all-pairs distances} privately with low error.

Interestingly, all of our error bounds are independent of the sum of the edge weights $\|w\|_1$, which corresponds
most naturally to the size of the database under the usual setting for differential privacy. Instead, they depend only
on the number of vertices $|\V|$ and edges $|\E|$ in the graph and the privacy parameters $\epsilon$ and $\delta$. 
Our error bounds constitute additive error. Consequently, if the edge weights are large,  
the error will be small in comparison.


\paragraph{Approximate shortest paths}
We consider the problem of privately releasing an approximate \emph{shortest path}.
We provide a strong reconstruction based lower bound,
showing that in general it is not possible under differential privacy to release a short path between a pair of vertices 
with additive error better than $\Omega(|\V|)$. 
Our lower bound is obtained by reducing the problem of reconstructing
many of the rows of a database to the problem of finding a path with low error. 

We also show that a simple $\epsilon$-differentially private algorithm based on the Laplace mechanism of  \cite{DMNS06}
comes close to meeting this bound.
If the minimum-weight path between a pair of vertices $s,t$ consists of at most $k$ edges, 
then the weight of the path released by
our algorithm is greater by at most $O(k\log |\V|) / \epsilon$.
Since $k < |\V|$, the weight of the released path is $O(|\V|\log |\V|)/\epsilon$ greater than optimal. 
Note that when the edge weights are large, the length of a path can be much larger than $(|\V|\log|\V|)/\epsilon$,
in which case our algorithm provides a good approximation.
Moreover, for many networks we would expect that shortest paths should be compact and consist of few vertices,
in which case the accuracy obtained will be much greater.
The algorithm releases not only a single path but short
paths between every pair of vertices without loss of additional privacy or accuracy.


\paragraph{Approximate all pairs distances}
For the problem of privately releasing all-pairs \emph{distances}, standard techniques yield error
$O(|\V|\log|\V|)/\epsilon$ for each query under $\epsilon$-differential privacy. 
We obtain improved algorithms for two special classes of graphs, 
trees and arbitrary graphs with edges of bounded weight.

For trees we show that a simple recursive algorithm can release all-pairs distances with error $O(\log^{2.5} |\V|) / \epsilon$.
Implicitly, the algorithm finds a collection $\C$ of paths with two properties. Every edge is contained in at most $\log |\V|$ paths,
and the unique path between any pair of vertices can be formed from at most $4\log |\V|$ paths in $\C$ using
the operations of concatenation and subtraction. The first property implies that the query consisting of the lengths of 
all of the paths in $\C$ has 
global $\ell_1$ sensitivity $\log |\V|$, so we can release estimates of the lengths of all paths in $\C$
with error roughly $\log |\V|$ using the standard Laplace mechanism \cite{DMNS06}.  The second property implies that we can use these estimates to compute
the distance between any pair of vertices without too much increase in error.
We first reduce the approximate all-pairs distances problem to the problem of approximating all distances from a single fixed vertex. 
We then solve the single-source problem recursively, repeatedly decomposing the tree into subtrees of at most half the size. 

In the special case of the path graph, releasing approximate all-pairs distances is
equivalent to query release of threshold
functions on the domain $X=\E$. 
The results of \cite{DNPR10} yield the same error 
bound that we obtain for computing distances on the path graph. Consequently our algorithm for all-pairs distances 
on trees can be viewed as a generalization of
a result for private query release of threshold functions. 

For bounded-weight graphs we show that we can generate a set of vertices $\Z\subset\V$ such that
distances between all pairs of vertices in $\Z$ suffice to estimate distances between all pairs of vertices
in $\V$. The required property on $\Z$ is that every vertex $v\in \V$ is connected to a vertex $z_v\in \Z$ by a path
consisting of few vertices.
Since we can always find a set $\Z$ that is not too large, we can release distances between all pairs of 
vertices in $\Z$ with relatively small error. The distance between a pair of vertices $u,v\in\V$ can then be approximated
by the distance between nearby vertices $z_u, z_v\in \Z$. 
In general, if edge weights are in $[0,\Lambda]$ for $\frac{1}{|\V|} < \Lambda\epsilon < |\V|$, then 
we can release all pairs shortest paths with an additive
error of
$\tilde{O}(\sqrt{|\V|\Lambda \epsilon^{-1}\log(1/\delta)})$ per path under $(\epsilon,\delta)$-differential privacy for any 
$\epsilon,\delta>0$. 
Note that the distance between a pair of vertices can be as large as $|\V|\cdot \Lambda$, so this error bound is nontrivial.
For specific graphs we can obtain better bounds by finding a smaller set $\Z$ satisfying the necessary requirements.

\paragraph{Additional problems}
We also consider two other problems in this model, the problem of releasing a nearly \emph{minimal spanning
tree} and the problem of releasing an almost \emph{minimum weight perfect matching}. 
For these problems, the \emph{approximation error} is the absolute difference in weight between the released
spanning tree or matching and the minimum-weight spanning tree or matching. 

Using a similar argument to the shortest path results of Section \ref{sec:DPSP}, we provide lower bounds
and algorithms for these problems. 
Through a reduction from the problem of reconstructing many rows in a database, we show that 
it is not possible under differential privacy to release a spanning tree or matching with error
better than $\Omega(|\V|)$. Using a simple algorithm based on the Laplace mechanism of \cite{DMNS06},
we can privately release a spanning tree or matching with error roughly $(|\V|\log |\V|)/\epsilon$. 
These results are presented
in Appendix \ref{sec:MST-matching-more}.

\paragraph{Future directions}
In this work we describe differentially private algorithms and lower bounds for several natural graph problems
related to network routing. 
We would like to explore how well our algorithms perform in practice on actual road networks and traffic data.  
We would also like to develop new algorithms for additional graph problems and to design improved all-pairs
distance algorithms for additional classes of networks.
In addition, lower bounds for accuracy of private all-pairs distances would be very interesting.

\paragraph{Scaling}
Our model requires that we preserve the indistinguishability of two edge weight functions which differ 
by at most $1$ in $\ell_1$ norm. 
The constant $1$ here is arbitrary, and the error bounds in this paper scale according to it. 
For instance, if a single individual can only influence edge weights by $1/|\V|$ rather than $1$ in $\ell_1$ norm, 
then 
we can privately
find a path between any pair of vertices whose length is only $O(\log|\V|)/\epsilon$ longer
than optimal rather than $O(|\V|\log|\V|)/\epsilon$. The other results in this paper can be scaled similarly.


\subsection{Previous work on graphs and differential privacy}
\label{ssec:RelatedWork}
\paragraph{Why a new model?}
The two main models which have been considered for applying differential privacy to network structured
data are \emph{edge} and \emph{node} differential privacy \cite{HLMJ09, BBDS13, KNRS13, KRSY11}. 
Under edge differential privacy, two graphs are considered to be neighbors if one graph can be obtained
from the other by adding or removing a single edge.
With node differential privacy, two graphs are neighbors if they differ only on the set of edges incident
to a single vertex.

However, the notions of edge and node differential privacy are not well suited to shortest
path and distance queries. Consider an algorithm that releases a short path between a specified pair of
vertices. Any useful program solving
this problem must usually at least output a path of edges which are in the graph.
But doing so blatantly violates privacy under both the edge and node notions of differential privacy,
since the released edges are private information. Indeed, an algorithm
cannot release subgraphs of the input graph without violating both edge
and node differential privacy
because this distinguishes the input from a neighboring graph in which one of the 
released edges is removed.

What if we simply want to release the distance between a pair of vertices?
In general it is not possible to release
approximate distances with meaningful utility under edge or node differential privacy,
since changing a single edge can greatly change the distances in the graph.
Consider the unweighted path graph $\P$ and any pair of adjacent vertices $x,y$ on the path. 
Removing edge $(x,y)$ disconnects the graph, increasing the distance between $x$ and $y$
from $1$ to $\infty$. 
Even if the graph remains connected, the removal of an edge does not preserve approximate distances.
Consider any pair of adjacent vertices $x,y$ on the cycle graph $\C$. Here, removing edge
$(x,y)$ increases the distance between $x$ and $y$ from $1$, the smallest possible distance,
to $|\V|-1$, the largest possible distance. 
Hence, the edge and node notions of differential privacy do not enable the release of approximate
distances.
This inadequacy motivates the new notion of differential privacy for graphs introduced in this work.

\paragraph{A histogram formulation}
A database consisting of elements from a data universe $\U$ can be described by a vector  
$D\in \NN^{|\U|}$, where the $i$th component of the vector corresponds to the number of copies of the $i$th element
of $\U$ in the database. 
More generally, we can allow the database to be
a point in $\RR^{|\U|}$, which corresponds to allowing a fractional number of occurrences of each element.
Since an edge weight function $w$ is an element of $\RR^{|\E|}$ and the notions of sensitivity coincide, 
we can rephrase our model in the standard differential privacy framework.  

Consequently, we can use existing tools for privately answering low sensitivity queries to yield incomparable results to
those presented in this paper for the problem of releasing all-pairs distances with low error. 
If all edge weights are integers and the sum $\|w\|_1$ of the edge weights is known,
we can release all-pairs distances with additive error 
$$\tilde{O}( \sqrt{ \wone\cdot \log |\V|} \log^{1.5}(1/\delta)/ \epsilon)$$
except with probability $\delta$
using the $(\epsilon, \delta)$-differentially private boosting mechanism of \cite{DRV10}.
The assumption that $\|w\|_1$ is known is not problematic, since we can privately release a good approximation. 
We can also extend the algorithm to noninteger edge weights for the queries in our setting, 
although we obtain a worse error bound. We do this by modifying the base synopsis generator of \cite{DRV10} to produce a database in
$(\tau \NN)^{|\E|}$ whose fractional counts are integer multiples of $\tau = \alpha / (2 |\V|)$, where $\alpha$ will be
the obtained additive approximation. The modified algorithm releases all-pairs distances with error 
$\tilde{O}(\sqrt[3]{\wone\cdot |\V|} \log^{4/3}(1/\delta) / \epsilon^{2/3})$.  

This error bound has a better dependence on $|\V|$ than the algorithms for all-pairs distances in general and 
weighted graphs described in Section \ref{sec:DPSdists}. However, it has a substantial dependence on
$\wone$, while the error bound for all algorithms in the remainder of this paper are independent of $\wone$. 
In addition, the running time of the algorithm of \cite{DRV10} is exponential in the database size, while
all algorithms described below run in polynomial time. 

\paragraph{Additional related work}
While the private edge weight model explored in this work is new, 
a few previous works have considered problems on graphs in related models. 
Nissim, Raskhodnikova and Smith 
\cite{NRS07} consider the problem of computing the cost 
of the minimum spanning tree 
of a weighted graph. They introduce the notion of smooth sensitivity, which they use 
to compute the approximate MST cost. In their work, edge weights are bounded, and
each weight corresponds to the data of a single individual. 
In contrast, we allow unbounded weights but assume a bound on the effect an individual can
have on the weights.

Hsu et al.~\cite{HHR+14} consider the problem of privately finding high-weight matchings
in bipartite graphs.
In their model, which captures a private version of the allocation problem, 
two weightings of the complete bipartite graph are neighbors if they differ only
in the weights assigned to the edges incident to a single left-vertex, which correspond
to the preferences of a particular agent. 
They show that the problem is impossible to solve under the standard notion of differential privacy.
They work instead under the relaxed notion of ``joint differential privacy,'' in which 
knowledge of the edges of the matching which correspond 
to some of the left-vertices cannot reveal the weights 
of edges incident to any of the remaining left-vertices.

A series of works has explored the problem of privately computing the size of all cuts
 $(S,\overline{S})$ in a graph. Gupta, Roth and Ullman
\cite{GRU12} show how to answer cut queries with $O(|\V|^{1.5})$ error.
Blocki et al.~\cite{BBDS12} improve the error for small cuts. 
Relatedly, Gupta et al.~\cite{GLMRT10} show that we can privately release a cut of close to minimal size with 
error $O(\log |\V|)/\epsilon$, and that this is optimal. 
Since the size of a cut is the number of edges crossing the cut, it can also be viewed as
the sum of the weights of the edges crossing the cut in a $\binset$-weighting of the complete graph.
Consequently the problem can be restated naturally in our model. 

As we have discussed, the problem of approximating all threshold functions on a totally ordered set
is equivalent to releasing approximate distances on the path graph. In addition to the work of 
Dwork et al.~\cite{DNPR10} described above, additional bounds and reductions
for query release and learning of threshold functions are shown in Bun et al.~\cite{BNSV15}.\\



\section{The privacy model}
\label{sec:prelim}
Let $\G=(\V,\E)$ denote an undirected graph with vertex set $\V$ and edge set $\E$, and let
$w:\E\to \RR^+$ be a \emph{weight function}. (The shortest path results in Section \ref{sec:DPSP}
also apply to directed graphs.)
Let $V=|\V|$ and $E=|\E|$ be the number of vertices and edges, respectively.

Let $\P_{xy}$ denote the set of paths between a pair of vertices $x,y\in\V$.
For any path $P\in \P_{xy}$,
the weight $w(P)$ is the sum $\sum_{e\in P} w(e)$ of the weights of the edges of $P$.
The distance $d_w(x,y)$ from $x$ to $y$ denotes the weighted distance $\min_{P\in \P_{xy}} w(P)$.
We will denote the unweighted or hop distance from $x$ to $y$ by
$\hd(x,y) = \min_{P\in \P_{xy}} \hl(P)$, where the hop length $\hl(P)$ of path 
$P=(v_0, \ldots, v_\ell)$ is the number $\ell$ of edges on the path. 
Let the shortest path $\SP_w(x,y)$ denote a path from $x$ to $y$ of minimum possible weight
$w(\SP_w(x,y))=d_w(x,y)$. 

We now formally define differential privacy in the private edge weight model.

\begin{defn}
For any edge set $\E$, two weight functions $w,w':\E\to\mathbb{R}^+$ are \textbf{neighboring}, denoted $w\sim w'$,
if $$\|w-w'\|_1 = \sum_{e\in \E} | w(e) - w'(e) | \leq 1.$$
\end{defn}

\begin{defn} 
For any graph $\G=(\V,\E)$, let $\A$ be an algorithm that takes as input a weight function $w:\E\to\mathbb{R}^+$. 
$\A$ is 
\textbf{$(\epsilon,\delta)$-differentially private} on $\G$ if 
for all pairs of neighboring weight functions $w,w'$
and for all sets of possible outputs $S$, we have that
$$\Pr[\A(w)\in S] \leq e^{\epsilon} \Pr[\A(w')\in S] + \delta.$$
If $\delta=0$ we say that $\A$ is $\epsilon$-differentially private on $\G$.
\end{defn}
For a class $\C$ of graphs, we say that $\A$ is $(\epsilon,\delta)$-differentially private on $\C$ if 
$\A$ is $(\epsilon,\delta)$-differentially private on $\G$ for all $\G\in\C$.

We now define our accuracy criteria for the approximate shortest paths and distances problems.
\begin{defn}
For the shortest path problem, the \textbf{error} of a path $P\in\P_{xy}$ between vertices 
$x,y$ is the difference
$w(P) - d_w(x,y)$ between the length of $P$ and the length of the shortest path between $x$ and $y$.
\end{defn}

\begin{defn}
For the approximate distances problem, the \textbf{error} is the absolute difference between
the released distance between a pair of vertices $x,y$ and the actual distance $d_w(x,y)$.
\end{defn}

\section{Preliminaries}

We will now introduce a few basic tools which will be used throughout the remainder of this paper.
A number of differential privacy techniques incorporate noise sampled according to the Laplace distribution.
We define the distribution and state a concentration bound for sums of Laplace random variables.

\begin{defn} 
The \textbf{Laplace distribution} with scale $b$, denoted $\Lap(b)$, is the distribution with probability density function
given by
$$p(x)=\frac{1}{2b}\exp(-|x|/b).$$
If $Y$ is distributed according to $\Lap(b)$, then for any $t>0$ we have that $\Pr[|Y|>t\cdot b] = e^{-t}$. 
\end{defn}

\begin{lemma} [Concentration of Laplace random variables \cite{CSS10}]
Let $X_1,\ldots,X_t$ be independent random variables distributed according to $\Lap(b)$,
and let $X=\sum_i X_i$.  
Then for all $\gamma\in (0,1)$ we have that with probability at least $1-\gamma$,
$$|X| < 4b\sqrt{t}\ln(2/\gamma) = O(b\sqrt{t}\log(1/\gamma)).$$
\label{lemma:LapSumConcBd}
\end{lemma}

One of the first and most versatile differentially private algorithms is the Laplace mechanism, which
releases a noisy answer with error sampled from the Laplace distribution with scale proportional to
the sensitivity of the function being computed.

\begin{defn}
For any function $f:\X\to\RR^k$, the \textbf{sensitivity} 
$$\Delta f = \max_{\overset{w,w'\in\X}{w\sim w'}} \| f(w) - f(w')\|_1$$
is the largest amount $f$ can differ in $\ell_1$ norm between neighboring inputs.
\end{defn}
In our setting we have $\X=(\RR^+)^\E$.

\begin{lemma}[Laplace mechanism \cite{DMNS06}]
Given any function $f:\X\to \RR^k$, the Laplace mechanism on input $w\in\X$
independently samples $Y_1,\ldots, Y_k$ according to $\Lap(\Delta f/\epsilon)$ and outputs
$$\M_{f,\epsilon}(w)=f(w) + (Y_1,\ldots,Y_k).$$
The Laplace mechanism is $\epsilon$-differentially private.
\label{lemma:LapMech}
\end{lemma}




Finally, we will need the following results on the composition of differentially private algorithms. 

\begin{lemma} [Basic Composition, e.g. \cite{DKMMN06}]
For any $\epsilon$, $\delta\geq 0$,
the adaptive composition of $k$ $(\epsilon,\delta)$-differentially private mechanisms is 
$(k\epsilon, k\delta)$-differentially private.
\label{lemma:basic_composition}
\end{lemma}

\begin{lemma} [Composition \cite{DRV10, DR13}]
For any $\epsilon, \delta, \delta' \geq 0$, the adaptive composition of $k$ $(\epsilon, \delta)$-differentially
private mechanisms is $(\epsilon', k\delta+\delta')$-differen-tially private for
$$\epsilon' = \sqrt{2k\ln(1/\delta')}\cdot\epsilon + k\epsilon(e^{\epsilon}-1)$$
which is $O(\sqrt{k\ln(1/\delta')}\cdot\epsilon)$
provided $k \leq 1/\epsilon^2$.
In particular, if $\epsilon' \in(0,1), \delta,\delta'\geq 0$, the composition of $k$ $(\epsilon, \delta)$-differentially private
mechanisms is $(\epsilon', k\delta+\delta')$-differentially private 
for $\epsilon = O(\epsilon' / \sqrt{k\ln(1/\delta')})$. 
\label{lemma:adv_composition}
\end{lemma}

\section{Computing distances}
\label{sec:DPSdists}


In this section we consider the problem of releasing distance oracle queries in the private edge weights model.
Since neighboring weight functions differ by at most $1$ in $\ell_1$ norm, 
the weight of any path also changes by at most 1.
Consequently, a single distance oracle query is sensitivity-1, and so we can use the Laplace mechanism
(Lemma \ref{lemma:LapMech})
to answer it privately after adding noise proportional to $1/\epsilon$. 
However, what if we would like to learn all-pairs distances privately?

There are $V^2$ pairs $(s,t)$ of vertices, so we can achieve $\epsilon$-differential privacy
by adding to each query Laplace noise proportional to $V^2/\epsilon$.
We can do better using approximate differential privacy ($\delta > 0$) and
Lemma \ref{lemma:adv_composition} (Composition). 
Adding Laplace noise proportional to $1/\epsilon'$
to each query results in a mechanism that is 
$(V\epsilon' \sqrt{2\ln(1/\delta)} + V^2\epsilon'(e^{\epsilon'}-1),\delta)$-differentially private
for any $\delta>0$. Taking $\epsilon'=\epsilon/O(V\sqrt{\ln 1/\delta})$ for $\epsilon<1$, we obtain a mechanism that is
$(\epsilon,\delta)$-differentially private by adding $O(V\sqrt{\ln 1/\delta})/\epsilon$ noise to each query.


The other natural approach is to release an $\epsilon$-differentially private 
version of the graph by adding  $\Lap(1/\epsilon)$ noise to each edge.
This will be the basis for our approach in Algorithm \ref{alg:SP} for computing approximate shortest paths.
With probability $1-\gamma$, all $E$ Laplace random variables will
have magnitude $(1/\epsilon)\log (E/\gamma)$, so
the length of every path in the released synthetic graph is within $(V/\epsilon)\log (E/\gamma)$
of the length of the corresponding path in the original graph.
Therefore with probability $1-\gamma$ we have that all pairs distances in the released synthetic
graph will be within $(V/\epsilon) \log (E/\gamma)$ of the corresponding distances in the original graph.

Both of these approaches result in privately releasing all pairs distances with additive error
roughly $V/\epsilon$. It is natural to ask whether this linear dependance on $V$ is the best possible.
In the remainder of this section we present algorithms which substantially improve on this 
bound for two special classes of graphs, trees and arbitrary graphs with edges of bounded weight. 
For trees we can obtain all-pairs distances with error $\polylog(V)/\epsilon$,
and for graphs with edges in $[0,M]$ we can obtain all-pairs distances with error roughly $\sqrt{VM/\epsilon}$. 

\subsection{Distances in trees}
\label{ssec:DPSdists_trees}


For trees it turns out to be possible to release all-pairs distances with far less error. 
We will first show a single-source version of this result for rooted trees, 
which we will then use to obtain the full result.
The idea is to split the tree into subtrees of at most half the size
of the original tree. As long as we can release the distance from the root to each subtree
with small error, we can then recurse on the subtrees. 

The problem of privately releasing all-pairs distances for the path graph can be restated as the problem of privately releasing
threshold functions for a totally ordered data universe of size $E=V-1$. 
Dwork et al.~\cite{DNPR10} showed that we can privately maintain a running counter for $T$ timesteps, 
with probability
$1-\gamma$ achieving additive accuracy of $O(\log(1/\gamma) \log^{2.5} T / \epsilon)$ at each timestep.
This algorithm essentially computes all threshold queries for a totally ordered universe of size $T$,
which is equivalent to computing the distance from an endpoint to each other vertex of the path graph.
We match the accuracy of the \cite{DNPR10} algorithm and generalize the result to arbitrary trees.  

In Appendix \ref{sec:ap-dists-path} we provide an alternate algorithm for the path graph which achieves the same 
asymptotic bounds. This result can be viewed as a restatement of the \cite{DNPR10} algorithm.

\begin{theorem} [Single-source distances on rooted trees]
Let $\T=(\V,\E)$ be a tree with root vetex $v_0$, and let $\epsilon>0$. Then
there is an algorithm that is $\epsilon$-differentially private on $\T$ that on input
edge weights $w:\E\to\RR^+$
releases approximate distances from $v_0$ to each other vertex,
where with probability $1-\gamma$ the approximation error on each released distance is
$$O(\log^{1.5}V \cdot\log (1/\gamma)) / \epsilon$$ for any $\gamma\in(0,1)$.
\label{thm:SSdists-trees}
\end{theorem}

\begin{figure*}[t]
\begin{center}
\begin{tikzpicture}[mn/.style={circle,fill=gray!20, minimum size=.7cm}]

\draw (0,0) node[mn] (v_0) {$v_0$} -- (1.5,1.5) node[mn] (a1) {};
\draw (v_0) -- (1.5,0) node[mn] (a2) {};
\draw (v_0) -- (1.5,-1.5) node[mn] (a3) {};

\draw (a1) -- (3,1.5) node[mn,fill=white] (e1) {$\cdots$};
\draw (a2) -- (3,0) node[mn,fill=white] (e2) {$\cdots$};
\draw (a3) -- (3,-1.5) node[mn,fill=white] (e3) {$\cdots$};

\draw (e3) -- (4.5,-1.5) node[mn] (v^*) {$v^*$};
\draw (v^*) -- (6.5,1.125) node[mn] (v_1) {$v_1$};
\draw (v^*) -- (6.5,-.75) node[mn] (v_2) {$v_2$};
\draw (6.5,-1.8) node[mn,fill=white] (v_dots) {$\vdots$};
\draw (v^*) -- (6.5,-3.05) node[mn] (v_t) {$v_t$};

\draw (v_1) -- (8,.75) node[mn,fill=white] (f_1) {$\cdots$};
\draw (v_2) -- (8,-.75) node[mn,fill=white] (f_2) {$\cdots$};
\draw (v_t) -- (8,-3.05) node[mn,fill=white] (f_t) {$\cdots$};

\draw (v_1) -- (8,1.5) node[mn,fill=white] (f_1prime) {$\cdots$};

\definecolor{mediumblue}{rgb}{0.0, 0.0, 0.8}
\node [draw=mediumblue, fit= (v_0) (v^*) (e1), inner sep=0.2cm, dashed] (T0) {};
\node [draw=mediumblue, fit= (v_1) (f_1) (f_1prime), inner sep=0.2cm, dashed] (T1) {};
\node [draw=mediumblue, fit= (v_2) (f_2), inner sep=0.2cm, dashed] (T2) {};
\node [draw=mediumblue, fit= (v_t) (f_t), inner sep=0.2cm, dashed] (Tt) {};

\node [xshift=-2.0ex, mediumblue] at (T0.west) {$\mathcal{T}_0$};
\node [xshift=-2.0ex, mediumblue] at (T1.west) {$\mathcal{T}_1$};
\node [xshift=-2.0ex, mediumblue] at (T2.west) {$\mathcal{T}_2$};
\node [xshift=-2.0ex, mediumblue] at (Tt.west) {$\mathcal{T}_t$};

\end{tikzpicture}
\end{center}
\caption{The partition used in Algorithm \ref{alg:tree-dists} to separate a tree into subtrees of size at most $V/2$.}
\label{fig:treepartition}
\end{figure*}
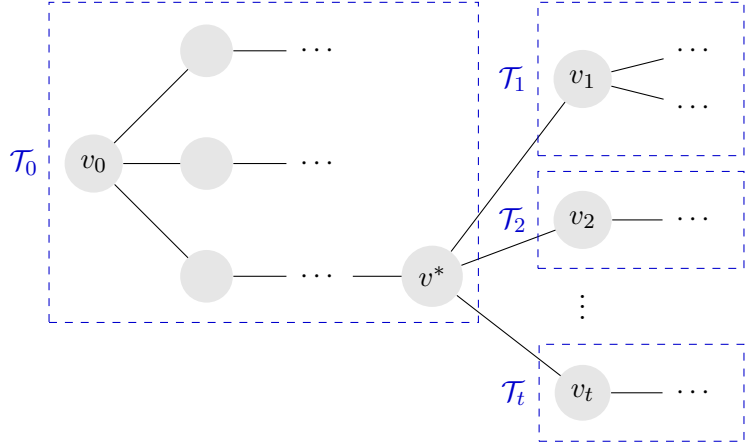

\begin{proof}
Given the tree $\T$ and root $v_0$, we will partition $\V$ into subtrees of size at most $V/2$ as shown in Figure 
\ref{fig:treepartition} and recursively
obtain distances in each subtree. 
There exists some vertex $v^*$ such that the subtree rooted at $v^*$ 
contains more than $V/2$ vertices while the subtree rooted at each child of $v^*$
has at most $V/2$ vertices. The topology of the graph is public, so we can release vertex $v^*$. 
Let $v_1, \ldots, v_t$ be the children of $v^*$, and let $\T_i=(\V_i,\E_i)$ be the subtree rooted at $v_i$
for each $i\in[t]$. Let $\T_0=(\V_0,\E_0)$ be the subtree rooted at $v_0$ consisting of the remaining vertices
$\V\setminus (\V_1\cup\cdots\cup \V_t)$. 

We can compute and release distances between the following 
pairs of vertices, adding $\Lap(\log V / \epsilon)$ noise to each distance:\\
\begin{itemize}
\item $(v_0,v^*)$
\item $(v^*, v_i)$ for each $i\in [t]$
\end{itemize}
Since $v^*$ is the parent of $v_1,\ldots,v_t$ in the tree rooted at $v_0$, the path from
$v_0$ to $v^*$ in $T$ contains none of the edges $(v^*,v_i)$ for $i\in[t]$, so the function
which releases these $t+1$ distances is sensitivity-1. 

We then recursively repeat the procedure on each subtree $\T_0,\ldots, \T_t$ until we reach
trees containing only a single vertex, adding $\Lap(\log V / \epsilon)$ noise to each released value. 
Each subtree has size at most $V/2$, so the depth of recursion is bounded by $\log V$.
The subtrees $\T_0,\ldots, \T_t$ are also disjoint. Consequently the function which releases 
all of the distances at depth $d$ in the recursion has sensitivity $1$ for any $d$. 
Therefore the function which releases all distances queried in this recursive procedure has
sensitivity at most $\log V$. Since we add $\Lap(\log V / \epsilon)$ noise to each coordinate
of this function, the algorithm outlined above is an instantiation of the Laplace mechanism (Lemma \ref{lemma:LapMech})
and is $\epsilon$-differentially private. 

\begin{figure*}[t]
\begin{algorithm}[Rooted tree distances]
\label{alg:tree-dists}
\noindent
\Inputs{Tree $\T=(\V,\E)$, root $v_0\in\V$, edge weights $w:\E\to\mathbb{R}^+$, and number of vertices $n$ of original tree.}
\begin{enumerate}
\item
Let $v^*$ be the unique vertex such that the subtree rooted at $v^*$ has more than $V/2$ vertices
while the subtree rooted at each child of $v^*$ has at most $V/2$ vertices.
\item Let $v_1,\ldots,v_t$ be the children of $v^*$.
\item Let $\T_i$ be the subtree rooted at $v_i$ for $i\in [t]$, and let $\T_0 = \T \setminus \{\T_1,\ldots,\T_t\}$.
\item Sample $X_{v^*,\T} \leftarrow  \Lap(\log V/\epsilon)$ and let $d(v^*,\T) = d_w(v_0,v^*) + X_{v^*,\T}$.
\item Let $d(v_0,\T)=0$. 
\item For each $i\in [t]$, sample $X_{v_i, \T} \leftarrow  \Lap(\log V/\epsilon)$ and
let $d(v_i,\T) = d(v^*,\T) + w((v^*,v_i)) + X_{v_i, \T}$. 
\item Recursively compute distances in each subtree $\T_0,\ldots,\T_t$. 
\item For each vertex $v\in \V$, if $v\in \T_i$, then let $d(v,\T) = d(v_i,\T) + d(v,\T_i)$. 
\item Release $d(v,\T)$ for all $v\in\V$. 
\end{enumerate}
\end{algorithm}
\end{figure*}

We now show how these queries suffice to compute the distance from root vertex $v_0$ to each other vertex with small error. 
The algorithm samples at most $2V$ Laplace random variables distributed according
to $\Lap(\log V / \epsilon)$, so by a union bound, with probability $1-\gamma$ all of these have magnitude 
$O(\log V \log (V/\gamma)) / \epsilon$. 
Consequently to obtain an error bound of roughly $O(\log^3 V) / \epsilon$
it suffices to show that any distance in $\T$ can be represented as a sum of $O(\log V)$ 
of the noisy distances released in the algorithm. We will use Lemma
\ref{lemma:LapSumConcBd} to obtain a slightly better bound on the error terms.

Let set $Q$ consist of the pairs of vertices corresponding to distance queries made by the algorithm above. 
We prove the following statement by induction. For any vertex $u\in \V$, there is a path from $v_0$ to $u$
in the graph $(\V,Q)$ consisting of at most $2\log V$ edges. The base case $V=1$ is vacuous. 
For $V>1$, note that since subtrees $T_0,\ldots, T_t$ partition the vertex set, $u$ must lie in one of them.
If $u\in T_i$, then by the inductive assumption on $T_i$, there is a path from $v_i$ to $u$ in $(V,Q)$
consisting of at most $2 \log(V/2) = 2 \log V - 2$ edges. If $i=0$ this already suffices. Otherwise,
if $i>0$, note that $(v_0,v^*)\in Q$ and $(v^*,v_i)\in Q$. Consequently 
there is a path in $(\V,Q)$ from $v_0$ to $u$ consisting of at most $2\log V$ edges.

This means that there is a set of at most $2\log V$ distances queried such that the distance from $v_0$ to
$u$ consists of the sum of these distances.
Consequently by adding these distances we obtain an estimate for the distance from $v_0$ to $u$ whose
error is the sum of at most $2\log V$ independent random variables each distributed according
to $\Lap(\log V/\epsilon)$.
By Lemma \ref{lemma:LapSumConcBd}, we have that with probability at least $1-\gamma$,
we compute an estimate of $d_w(v_0,u)$ with error $O(\log^{1.5} V \cdot \log(1/\gamma)) / \epsilon$
for any $\gamma\in(0,1)$. 
Since differentially private mechanism are preserved under post-processing, this algorithm
satisfies $\epsilon$-differential privacy and computes
distances from $v^*$ to each other vertex
in $T$, where with probability at least $1-\gamma$
each distance has error at most $O(\log^{1.5} V\cdot \log(1/\gamma)) /\epsilon$.
\end{proof}

We now extend this result to obtain all-pairs distances for trees.

\begin{theorem} [All-pairs distances on trees]
For any tree $\T=(\V,\E)$ and $\epsilon>0$ there is an algorithm that is $\epsilon$-differentially private on $\T$
and on input edge weights $w:\E\to\mathbb{R}^+$
releases all-pairs distances such that with probability $1-\gamma$, all released distances have
approximation error 
$$O(\log^{2.5} V\cdot \log(1/\gamma)) / \epsilon$$ 
for any $\gamma\in(0,1)$.
For each released distance, with probability $1-\gamma$ the approximation error is $O(\log^{1.5} V\cdot \log(1/\gamma)) / \epsilon.$
\label{thm:APdists-trees}
\end{theorem}
\begin{proof}
Arbitrarily choose some root vertex $v_0$. Use the $\epsilon$-differentially private
algorithm of Theorem \ref{thm:SSdists-trees} to obtain approximate distances from
$v_0$ to each other vertex of $\T$. 
We now show that this suffices to obtain all-pairs distances.

Consider any pair of vertices $x,y$, and let $z$ be their lowest common ancestor in the tree rooted at $v_0$.
Then 
\begin{align*}
d_w(x,y) &= d_w(z,x) + d_w(z,y)\\
 &= d_w(v_0,x) + d_w(v_0,y) - 2 d_w(v_0,z).
\end{align*}
Since with probability $1-3\gamma$ we can compute each of these three distances with error $O(\log^{1.5} V\cdot\log(1/\gamma)/\epsilon)$,
it follows that we can compute the distance between $x$ and $y$ with error
at most four times this, which is still $O(\log^{1.5} V\cdot \log(1/\gamma)) / \epsilon$. 
By a union bound, for any $\gamma\in(0,1)$, with probability at least $1-\gamma$
each error among the $V(V-1)/2$ all-pairs distances released is
at most 
$O(\log^{1.5} V \cdot\log(V/\gamma))/\epsilon = O(\log^{2.5} V \cdot\log(1/\gamma))/\epsilon$. 
\end{proof}

\subsection{Distances in bounded-weight graphs}
\label{ssec:DPSdists_bddwt}

\begin{theorem}
For all $\G, \delta,\gamma, \Lambda$, and $\epsilon\in(0,1)$, if $1/V<\Lambda\epsilon< V$
then there is an algorithm $\A$ 
that is $(\epsilon,\delta)$-differentially private on $\G$ such that for all
$w:\E\to[0,\Lambda]$, with probability $1-\gamma$, $\A(w)$ outputs 
approximate all-pairs distances with additive error 
$$O\left(\sqrt{V\Lambda\epsilon^{-1} \cdot \log (1/\delta)} \cdot\log(V\Lambda\epsilon/\gamma)\right)$$ 
per distance.
For any $\epsilon>0$, $\G,\gamma,\Lambda$, if $1/V<\Lambda\epsilon< V^2$ 
then there is an algorithm $\B$ that is $\epsilon$-differentially private on $\G$
such that for all $w:\E\to[0,\Lambda]$, with probability $1-\gamma$,
$\B(w)$ outputs approximate all-pairs distances with additive error 
$$O\left((V\Lambda)^{2/3}\epsilon^{-1/3}\log(V\Lambda\epsilon/\gamma)\right)$$ 
per distance.

\label{thm:apsp_bounded_wts}
\end{theorem}

\begin{figure*}[tb]
\begin{algorithm}[Bounded-weight distances]
\label{alg:bw-dists}
\noindent
\Inputs{$\G=(\V,\E)$, $w:\E\to\mathbb{R}^+$, $k,\Lambda,\gamma,\epsilon'>0$, $k$-covering $\Z$.}
\begin{enumerate}
\item For all $y,z\in\Z$, sample $X_{y,z} \leftarrow\Lap(|\Z|/\epsilon')$ and output $a_{y,z} := d_w(y,z) + X_{y,z}$. 
\item For $v\in V$, let $z(v)\in\Z$ denote a vertex in $\Z$ with $\hd(v,z(v)) \leq k$.
\item The approximate distance between vertices $u,v\in V$ is given by $a_{z(u), z(v)}$. 
\end{enumerate}
\end{algorithm}
\end{figure*}

To achieve this result, we will find a small subset $\Z$ of $\V$ such that every vertex $v\in \V$ is near
some vertex $z\in \Z$. We will need the following definition, introduced in \cite{MM75}.

\begin{defn}
A subset $\Z\subset \V$ of vertices is a \textbf{$k$-covering} if for every $v\in \V$ there is some 
$z\in \Z$ such that $\hd(v,z) \leq k$, where $\hd$ is the hop-distance. 
\end{defn}

A $k$-covering is sometimes called a $k$-dominating set (e.g.~in \cite{CN82}). 
The following lemma shows that we can find a sufficiently small $k$-covering for any graph $\G$. 


\begin{lemma}\cite{MM75}
If $V\geq k+1$, then $\G$ has a $k$-covering of size at most $\lfloor V / (k+1) \rfloor$.
\label{lemma:kcoveringsize}
\end{lemma}
\begin{proof}
Consider any spanning tree $\T$ of $\G$ and any vertex $x\in \V$ that is an endpoint of one of the
longest paths in $\T$. For $0\leq i\leq k$, 
let $\Z_i$ be the subset of $\V$ consisting of vertices whose distance from $x$ in $\T$ is
congruent to $i$ modulo $k+1$.
It can be shown that each $\Z_i$ is a $k$-covering of $\T$ and therefore of $\G$.
But the $k+1$ sets $\Z_i$ form a partition of $\V$. Therefore some $\Z_i$ has 
$|\Z_i| \leq  \lfloor V / (k+1) \rfloor$, as desired.
\end{proof}




\begin{theorem}
For any $\G=(\V,\E), k,\delta>0$, and $\gamma,\epsilon\in(0,1)$,
let $\Z$ be a $k$-covering of $\G$, and let $Z=|\Z|$. Then
there is an algorithm $\A$ which is $(\epsilon,\delta)$-differentially private on $\G$
such that for any edge weight function $w:\E\to [0,\Lambda]$,
with probability $1-\gamma$, $\A(w)$ releases 
all-pairs distances with approximation error 
$$O\left(k\Lambda+Z\epsilon^{-1}\log(Z/\gamma)\sqrt{\log(1/\delta)}\right)$$ 
per distance.

\label{thm:bddwt_covering_advanced}
\end{theorem}

\begin{proof}
There are $Z^2$ pairwise distances between vertices in $\Z$. 
We can compute and release noisy versions of each of these distances, 
adding $\Lap(Z/\epsilon')$ noise to each. For any $\gamma\in(0,1)$,
with probability $1-\gamma$ we have that each of these $Z^2$ noise
variables has magnitude at most $(Z/\epsilon')\log(Z^2/\gamma)$.
By Lemma \ref{lemma:adv_composition}, for any $\delta>0$ releasing
this is $(\epsilon,\delta)$-differentially private for $\epsilon'=O(\epsilon/ \sqrt{\ln\,1/\delta}\,)$.

But this information allows the recipient to compute approximate distances between any pair of vertices
$x,y\in \V$, as follows. Since $\Z$ is a $k$-covering of $\G$, we can find $z_x,z_y\in \Z$
which are at most $k$ vertices from $x$ and $y$. 
Since the maximum weight is $\Lambda$, the weight of the shortest path 
between $x$ and $z_x$ is at most $k\Lambda$, and similarly for $y$ and $z_y$.
Consequently 
$$\left| d_w(x,y) - d_w(z_x,z_y)\right| \leq 2 k\Lambda.$$
But we have released an estimate of $d_w(z_x,z_y)$ with noise distributed 
according to $\Lap(Z/\epsilon)$. Consequently with probability $1-\gamma$
each of these estimates differs from $d_w(z_x,z_y)$ by at most $(Z/\epsilon)\log(Z^2/\gamma)$. 
\end{proof}

We obtain a slightly weaker result under pure differential privacy. 


\begin{theorem}
Let $\G=(\V,\E)$.
For any $k>0$ and $\gamma\in(0,1)$, if $\Z$ is a 
$k$-covering of $\G$ of size $Z$, then there is an algorithm $\A$ that is 
$\epsilon$-differentially private on $\G$ 
such that for any $w:\E\to [0,\Lambda]$,
with probability $1-\gamma$, $\A(w)$ releases all-pairs
distances with approximation error 
$$O(k\Lambda+Z^2\epsilon^{-1}\log(Z/\gamma))$$
per distance.
\label{thm:bddwt_covering_pure}
\end{theorem}

\begin{proof}
There are at most $Z^2$ pairwise distances between vertices in $\Z$, so
we can release approximations of each distance, adding $\Lap(Z^2/\epsilon)$ noise to each distance.
With probability $1-\gamma$ each of these $Z^2$ noise variables has magnitude at most $(Z^2/\epsilon)\log(Z^2/\gamma)$. 
By Lemma \ref{lemma:basic_composition}, releasing these distances is $\epsilon$-differentially private.
As above, since $\Z$ is a $k$-covering of $\G$, we can find $z_x, z_y\in \Z$ which are at most
$k$ vertices from $x$ and $y$. 
Consequently $d_w(x,z_x)\leq k\Lambda$ and $d_w(y,z_y) \leq k\Lambda$,
so 
$$\left|d_w(x,y) - d_w(z_x,z_y)\right| \leq 2 k\Lambda .$$
But the released estimate for $d_w(z_x,z_y)$ has noise distributed according to $\Lap(k\Lambda)$,
so with probability $1-\gamma$ each of these estimates differs from $d_w(x,y)$ by
at most $O(Z^2\epsilon^{-1}\log(Z^2/\gamma))$, as desired. 
\end{proof}

We now conclude the proof of Theorem \ref{thm:apsp_bounded_wts}.

\begin{proof}[Proof of Theorem \ref{thm:apsp_bounded_wts}]
Combine Lemma \ref{lemma:kcoveringsize}
with Theorem \ref{thm:bddwt_covering_advanced} for $k=\lfloor \sqrt{V/(\Lambda\epsilon)}\,\rfloor$, and
combine Lemma \ref{lemma:kcoveringsize} with Theorem \ref{thm:bddwt_covering_pure} for 
$k=\lfloor V^{2/3}/(\Lambda\epsilon)^{1/3}\rfloor$. 
\end{proof}

Note that if we are only interested in the $V-1$ distances from a single source, then directly
releasing noisy distances and applying Lemma \ref{lemma:adv_composition} yields 
$(\epsilon,\delta)$-differential privacy
with error distributed
according to $\Lap(b)$ for $b=O(\sqrt{V\log 1/\delta}\,)/\epsilon$,
which has the same dependence on $V$ as the bound 
provided by the theorem for releasing all pairs distances.

For some graphs we may be able to find a smaller $k$-covering than
that guaranteed by Lemma \ref{lemma:kcoveringsize}.
Then we can use Theorems \ref{thm:bddwt_covering_advanced} and \ref{thm:bddwt_covering_pure}
to obtain all-pairs distances with lower error. 
For instance, we have the following.

\begin{theorem}
Let $\G$ be the $\sqrt{V} \times \sqrt{V}$ grid.
Then for any $\epsilon,\gamma\in(0,1)$ and $\delta>0$, we can release with probability $1-\gamma$ all-pairs distances 
with additive approximation error 
$$V^{1/3} \cdot O\left(\Lambda + \epsilon^{-1}\log (V/\gamma)\sqrt{\log 1/\delta}\,\right)$$
while satisfying
$( \epsilon,\delta)$-differential privacy.
\end{theorem}
\begin{proof}
Let $\V=[\sqrt{V}]\times[\sqrt{V}]$, and let $\Z\subset \V$ consist of vertices $(i,j)\in \V$ with $i,j$ both one less than 
a multiple of $V^{1/3}$.
Then $\Z$ is a $2V^{1/3}$-covering of $\G$, and also $|\Z| \leq V^{1/3}$.
Consequently Theorem \ref {thm:bddwt_covering_advanced} implies the desired conclusion.
\end{proof}

\section{Finding shortest paths}
\label{sec:DPSP}

\subsection{Lower bound}
\label{ssec:DPSPlower}

In this section we present a lower bound on the additive error with which we can privately release a short path between
a pair of vertices. The argument is based on a reduction from the problem of reconstructing a large fraction
of the entries of a database. We show that an adversary can use an algorithm which outputs a
short path in a graph to produce a vector with small Hamming distance to an input vector, which is impossible
under differential privacy. To that end, 
we exhibit a ``hard'' graph $\G=(\V,\E)$ and a family of weight functions, and provide a correspondence
between inputs $x\in\binset^n$ and weight functions $w:\E\to\binset$.  

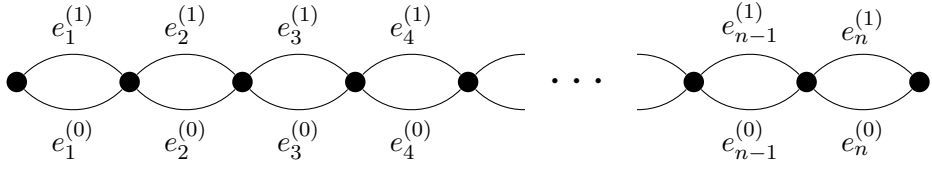
\begin{figure*}[t!]
\begin{center}
\begin{tikzpicture}[mn/.style={circle,fill=black, scale=.75}]

\draw (0,0) node[mn] (v0) {};
\draw (1.5,0) node[mn] (v1) {};
\draw (3,0) node[mn] (v2) {};
\draw (4.5,0) node[mn] (v3) {};
\draw (6,0) node[mn] (v4) {};

\draw (7.5,0) node[mn,fill=white] (v5) {};
\draw (7.5,0) node[mn,fill=white] (v5text) {\Huge{$\cdots$}};

\draw (9,0) node[mn] (v6) {};
\draw (10.5,0) node[mn] (v7) {};
\draw (12,0) node[mn] (v8) {};

\path
(v0) edge [bend right = 45] node [below] {$e_1^{(0)}$} (v1)
(v0) edge [bend left = 45] node [above] {$e_1^{(1)}$} (v1)

(v1) edge [bend right = 45] node [below] {$e_2^{(0)}$} (v2)
(v1) edge [bend left = 45] node [above] {$e_2^{(1)}$} (v2)

(v2) edge [bend right = 45] node [below] {$e_3^{(0)}$} (v3)
(v2) edge [bend left = 45] node [above] {$e_3^{(1)}$} (v3)

(v3) edge [bend right = 45] node [below] {$e_4^{(0)}$} (v4)
(v3) edge [bend left = 45] node [above] {$e_4^{(1)}$} (v4)

(v6) edge [bend right = 45] node [below] {$e_{n-1}^{(0)}$} (v7)
(v6) edge [bend left = 45] node [above] {$e_{n-1}^{(1)}$} (v7)

(v7) edge [bend right = 45] node [below] {$e_n^{(0)}$} (v8)
(v7) edge [bend left = 45] node [above] {$e_n^{(1)}$} (v8)
;

\begin{scope}
\clip (6,-0.5) rectangle (6.75,0.5);
\path
(v4) edge [bend right = 45] node {} (v5)
(v4) edge [bend left = 45] node {} (v5);
\end{scope}

\begin{scope}
\clip (8.25,-0.5) rectangle (9,0.5);
\path
(v5) edge [bend right = 45] node {} (v6)
(v5) edge [bend left = 45] node {} (v6);
\end{scope}

\end{tikzpicture}
\end{center}
\caption{The graph used for the lower bound of Lemma \ref{lemma:SP-reduction}.}
\label{fig:SP-lower-bound}
\end{figure*}

\begin{theorem}
There exists a graph $\G=(\V,\E)$ and vertices $s,t\in \V$ such that 
for any algorithm $\A$ that is $(\epsilon,\delta)$-differentially private on $\G$, there exist edge weights 
$w:\E\to\binset$ for which
the expected approximation error of the path $\A(w)$ from $s$ to $t$ is at least 
$\alpha=(V-1)\cdot\left( \frac{1-(1+e^{\epsilon})\delta}{1+e^{2\epsilon}}\right)$.
In particular, for sufficiently small $\epsilon$ and $\delta$, $\alpha \geq 0.49 (V-1)$. 
\label{thm:SP-lower-bound}
\end{theorem}

Consequently, any differentially private algorithm for approximate shortest paths must on some inputs have
additive error proportional to the number of vertices. 
Let $\G=(\V,\E)$ be the $(n+1)$-vertex graph with vertex set $\V=\{0,\ldots,n\}$
and two  parallel edges $e_i^{(0)}$ and $e_i^{(1)}$ between
each pair of consecutive vertices $i-1, i$, as shown in Figure \ref{fig:SP-lower-bound}.
(For simplicity we have defined this is a multigraph, but it can be transformed into a simple graph with the 
addition of $n$ extra vertices, changing the bound obtained by a factor of $2$.)

\begin{lemma}
\label{lemma:SP-reduction}
Let $\G=(\V,\E)$ be the graph defined in the previous paragraph. 
For any $\alpha$, let $\A$ be an algorithm
that is $(\epsilon,\delta)$-differentially private on $\G$ 
that on input edge weights $w: \E\to \{0,1\}$
produces a path from vertex $s=0$ to vertex $t=n$
with expected approximation error at most $\alpha$. 
Then there exists a
$(2\epsilon, (1+e^{\epsilon})\delta)$-differentially private algorithm $\B$ which on input $x\in\{0,1\}^n$ produces
$y\in \{0,1\}^n$ such that the expected Hamming distance $d_H(x,y)$ is at most $\alpha$.
\end{lemma}
\begin{proof}
Given an input $x\in \{0,1\}^n$, 
the corresponding edge weight function $w_x$ is given by $w_x(e_i^{(x_i)}) = 0$ and $w_x(e_i^{(1-x_i)}) = 1$. 
That is, for each pair of consecutive vertices $i-1, i$, one of the edges between them will have weight
0 and the other will have weight 1 as determined by the $i$th bit of the input $x$. 

The algorithm $\B$ is as follows. On input $x\in \{0,1\}^n$, apply $\A$ to $(\G,w_x)$, 
and let $\P$ be the path produced. 
Define $y\in \{0,1\}^n$ as follows. Let $y_i = 0$ if $e_i^{(0)}\in \P$ and $y_i=1$ otherwise. Output $y$.

We first show that this procedure is differentially private. Given neighboring inputs $x,x'\in \{0,1\}^n$
which differ only on a single coordinate $x_i\neq x_i'$, we have that the associated weight functions 
$w_x$ and $w_{x'}$ have $\ell_1$ distance 2, since they disagree only edges $e_i^{(0)}$ and $e_i^{(1)}$. Consequently,
since $\A$ is $(\epsilon,\delta)$-differentially private, we have that for any set of values $S$ in the range of $\A$,
\begin{align*}
\Pr[\A(w_x)\in S] &\leq e^{\epsilon} (e^{\epsilon} \Pr[\A(w_{x'})\in S] + \delta) + \delta\\
&= e^{2\epsilon} \Pr[\A(w_{x'})\in S] + (1+e^{\epsilon})\delta~.
\end{align*}
But algorithm $\B$ only accesses the database $x$ through $\A$. Consequently by the robustness of differential
privacy to post-processing, we have that $\B$ is $(2\epsilon, (1+e^{\epsilon})\delta)$-differentially private.

We now show that the expected number of coordinates in which $y$ disagrees with $x$ is at most $\alpha$.
The shortest path from $s$ to $t$ in $\G$ has length $0$, so the expected length of the path $\P$ produced by 
$\A$ is at most $\alpha$. Consequently in expectation $\P$ consists of at most $\alpha$ edges $e_i^{(1-x_i)}$. 
But $y_i \neq x_i$ only if  $e_i^{(1-x_i)} \in \P$, so it follows that the expected Hamming distance $d_H(x,y) \leq \alpha$.
\end{proof}

We will now prove two simple and standard lemmas concerning the limits of identifying rows of the input 
of a differentially private algorithm.  
In the first lemma, we show that a differentially private algorithm cannot release a particular
row of its input with probability greater than $\frac{e^{\epsilon}+\delta}{1+e^{\epsilon}}$.
In essence, for $\delta=0$ this can be interpreted as a statement about the optimality
of the technique of randomized response \cite{Warner65}.

\begin{lemma}
\label{lemma:reidentification-projection}
If algorithm $\B:\{0,1\}^n\to\{0,1\}$ is $(\epsilon,\delta)$-differentially private, then 
if we uniformly sample a random input $X\leftarrow U_n$, we have that for all $i$,
$$\Pr[\B(X) \neq X_i] \geq \frac{1 - \delta}{1+e^{\epsilon}}$$
\end{lemma}
\begin{proof}
We have that $\Pr[\B(X) = X_i] $ is given by
\begin{align*}
&\phantom{==}\frac{1}{2} \Pr[\B(X_{-i},0) = 0] + \frac{1}{2} \Pr[\B(X_{-i},1) = 1]\\
&\leq \frac{e^\epsilon}{2}  \cdot (\Pr[\B(X_{-i},1) = 0] +  \Pr[\B(X_{-i},0) = 1]) + \delta\\
&= e^\epsilon \Pr[\B(X) \neq X_i] + \delta
\end{align*}
so since $\Pr[\B(X) = X_i] = 1 -  \Pr[\B(X) \neq X_i] $,
$$\Pr[\B(X) \neq X_i] \geq \frac{1 - \delta}{1+e^{\epsilon}}~.$$
\end{proof}

This immediately implies the following result.

\begin{lemma}
\label{lemma:reidentification-Hamming}
If algorithm $\B:\{0,1\}^n\to\{0,1\}^n$ is $(\epsilon,\delta)$-differentially private, then
for some $x\in \{0,1\}^n$ we have that
the expected Hamming distance $d_H(\B(x),x)$ 
is at least 
$ \frac{n(1 - \delta)}{1+e^{\epsilon}}$.
\end{lemma}
\begin{proof}
By Lemma \ref{lemma:reidentification-projection}, projecting onto any coordinate $i$, the probability
that $\Pr[\B(X)_i \neq X_i]$ for uniformly random $X\leftarrow U_n$
is at least $\frac{1 - \delta}{1+e^{\epsilon}}$.
Consequently the expected number of coordinates on which
$\B(X)$ differs from $X$ is $\EE(d_H(\B(X),X)) \geq \frac{n(1 - \delta)}{1+e^{\epsilon}}$.
Since this holds for $X$ uniformly random, in particular we have that there exists
some $x\in \{0,1\}^n$ such that 
$$\EE(d_H(\B(x),x)) \geq \frac{n(1 - \delta)}{1+e^{\epsilon}}.$$
\end{proof}


We now conclude the proof of Theorem \ref{thm:SP-lower-bound}.

\begin{proof}[Proof of Theorem \ref{thm:SP-lower-bound}]
Assume that there is some algorithm which is $(\epsilon,\delta)$-differentially private on $\G$
and always produces 
a path of error at most $\alpha$ between $s$ and $t$.  Then by Lemma 
\ref{lemma:SP-reduction} we have that there is a
$(2\epsilon, (1+e^{\epsilon})\delta)$-differentially private algorithm $\B$ which for all $x\in\{0,1\}^n$ produces
$y\in \{0,1\}^n$ such that the expected Hamming distance $d_H(x,y)$ is less than $\alpha$.
But by Lemma \ref{lemma:reidentification-Hamming}, for some $x$
the expected Hamming distance $\EE(d_H(x,y)) \geq 
 \frac{n(1 - (1+e^{\epsilon})\delta)}{1+e^{2\epsilon}}=\alpha$, yielding a contradiction.
\end{proof}

\subsection{Upper bound}
\label{ssec:DPSPupper}

In this section we show that an extremely simple algorithm
matches the lower bound of the previous section
up to a logarithmic factor, for fixed $\epsilon,\delta$.
Consider a direct application of the Laplace mechanism (Lemma \ref{lemma:LapMech}), 
adding $\Lap(1/\epsilon)$ noise to each edge weight
and releasing the resulting values. With high probably all of these $E<V^2$ noise variables will be small,
providing a bound on the difference in the weight of any path between the released graph and the original graph.
Consequently we can show that if we take the shortest path in the released graph, with $99\%$ probability
the length of the same path in the original graph is $O(V\log V)/\epsilon$ longer than optimal.

This straightforward application of the Laplace mechanism almost matches the lower bound of the previous section.
Surprisingly, with the same error bound it releases not just a short path between a single pair of vertices
but short paths between all pairs of vertices.

One drawback of this argument is that the error depends on the size of the entire graph. In practice
we may expect that the shortest path between most pairs of vertices consist of relatively few edges. 
We would like the error to depend on the number of hops on the shortest path rather than scaling with the
number of vertices. We achieve this with a post-processing step that increases the weight of all edges,
introducing a preference for few-hop paths. We show that if there is a short path with only $k$ hops, then our
algorithm reports a path whose length is at most $O(k\log V/\epsilon)$ longer. 


\begin{theorem}
For all graphs $\G$ and $\gamma\in(0,1)$,
Algorithm \ref{alg:SP} is $\epsilon$-differentially private on $\G$ and
computes paths between all pairs of vertices such that with probability $1-\gamma$, for all pairs of vertices
$s,t\in \V$, if there exists a $k$-hop path of weight $W$ in $(\G,w)$,
the path released has weight at most $W+(2k/\epsilon)\log(E/\gamma)$.
\label{thm:SP-upper-bound}
\end{theorem}

In particular, if the shortest path in $\G$ has $k$ hops, then Algorithm \ref{alg:SP} 
releases a path only $(2k/\epsilon)\log(E/\gamma)$ longer than optimal. This error term is proportional to 
the number
of hops on the shortest path.
Noting that the shortest path between any pair of vertices consists of fewer than $V$ hops, 
we obtain the following corollary. 

\begin{corollary}
For any $\gamma\in(0,1)$, with probability $1-\gamma$ Algorithm \ref{alg:SP} computes paths between every pair of vertices
with approximation error at most $(2V/\epsilon)\log(E/\gamma)$.
\end{corollary}

\begin{figure*}[t]
\begin{algorithm}[Private shortest paths]
\label{alg:SP}
\noindent
\Inputs{$\G=(\V,\E)$, $w:\E\to\mathbb{R}^+$, $\gamma>0$.}
For each edge $e\in \E$, do the following:
\begin{enumerate}
\item Sample $X_e \leftarrow \Lap(1/\epsilon)$
\item Let $w'(e) = w(e) + X_e + (1/\epsilon)\log(E/\gamma)$, where $E=|\E|$
\end{enumerate}
Output $w'$.\\\\
The approximate shortest path between a pair of vertices $x,y\in\V$ is taken to be the shortest path
$\SP_{w'}(x,y)$ in the weighted graph $(\G,w')$. 
\end{algorithm}
\end{figure*}

\begin{proof}[Proof of Theorem \ref{thm:SP-upper-bound}]
Each random variable $X_e$ is distributed according to $\Lap(1/\epsilon)$, 
so with probability $1-\gamma$ we have that $|X_e| \leq (1/\epsilon)\log(1/\gamma)$
for any $\gamma\in(0,1)$. 
By a union bound, with probability $1-\gamma$ 
all $E<V^2$ of these random variables have magnitude at most
$(1/\epsilon)\log(E/\gamma)$. Conditioning on this event, for each edge $e\in\E$,
the modified weight computed by the algorithm satisfies
$$w(e)\leq w'(e) \leq w(e)+(2/\epsilon)\log(E/\gamma).$$
Therefore, for any $k$-hop path $\P$ we have that
$$w(\P)\leq w'(\P) \leq w(\P)+(2k/\epsilon)\log(E/\gamma).$$
For any $s,t\in \V$, if $\Q$ is the path from $s$ to $t$ produced by the algorithm and $\Q'$ is any path from $s$ to $t$,
then we have that $\Q$ is a shortest path under $w'$, so
$$w(\Q)\leq w'(\Q) \leq w'(\Q') \leq w(\Q') + (2\hl(\Q')/\epsilon)\log(E/\gamma).$$
\end{proof}

\paragraph{Acknowledgments}
I am grateful to Salil Vadhan and Shafi Goldwasser for their guidance and encouragement throughout this project. Thanks to Jon Ullman, Tommy MacWilliam and anonymous referees for helpful suggestions.
This work was supported by a DOE CSGF Fellowship, NSF Frontier CNS-1413920, DARPA W911NF-15-C-0236, 
and the Simons Foundation (agreement dated 6-5-12).



\bibliography{DP_SPsandDists}

\newcommand{\etalchar}[1]{$^{#1}$}
\begin{thebibliography}{DKM{\etalchar{+}}06}

\bibitem[BBDS12]{BBDS12}
Jeremiah Blocki, Avrim Blum, Anupam Datta, and Or~Sheffet.
\newblock The johnson-lindenstrauss transform itself preserves differential
  privacy.
\newblock In {\em Foundations of Computer Science (FOCS), 2012 IEEE 53rd Annual
  Symposium on}, pages 410--419. IEEE, 2012.

\bibitem[BBDS13]{BBDS13}
Jeremiah Blocki, Avrim Blum, Anupam Datta, and Or~Sheffet.
\newblock Differentially private data analysis of social networks via
  restricted sensitivity.
\newblock In {\em Proceedings of the 4th conference on Innovations in
  Theoretical Computer Science}, pages 87--96. ACM, 2013.

\bibitem[BNSV15]{BNSV15}
Mark Bun, Kobbi Nissim, Uri Stemmer, and Salil Vadhan.
\newblock Differentially private release and learning of threshold functions.
\newblock {\em arXiv preprint arXiv:1504.07553}, 2015.

\bibitem[CN82]{CN82}
GJ~Chang and GL~Nemhauser.
\newblock The k-domination and k-stability problem on graphs.
\newblock {\em Techn. Report}, 540, 1982.

\bibitem[CSS10]{CSS10}
TH~Hubert Chan, Elaine Shi, and Dawn Song.
\newblock Private and continual release of statistics.
\newblock In {\em Automata, Languages and Programming}, pages 405--417.
  Springer, 2010.

\bibitem[DKM{\etalchar{+}}06]{DKMMN06}
Cynthia Dwork, Krishnaram Kenthapadi, Frank McSherry, Ilya Mironov, and Moni
  Naor.
\newblock Our data, ourselves: Privacy via distributed noise generation.
\newblock In {\em Advances in Cryptology-EUROCRYPT 2006}, pages 486--503.
  Springer, 2006.

\bibitem[DMNS06]{DMNS06}
Cynthia Dwork, Frank McSherry, Kobbi Nissim, and Adam Smith.
\newblock Calibrating noise to sensitivity in private data analysis.
\newblock In {\em Theory of cryptography}, pages 265--284. Springer, 2006.

\bibitem[DNPR10]{DNPR10}
Cynthia Dwork, Moni Naor, Toniann Pitassi, and Guy~N Rothblum.
\newblock Differential privacy under continual observation.
\newblock In {\em Proceedings of the forty-second ACM symposium on Theory of
  computing}, pages 715--724. ACM, 2010.

\bibitem[DR13]{DR13}
Cynthia Dwork and Aaron Roth.
\newblock The algorithmic foundations of differential privacy.
\newblock {\em Theoretical Computer Science}, 9(3-4):211--407, 2013.

\bibitem[DRV10]{DRV10}
Cynthia Dwork, Guy~N Rothblum, and Salil Vadhan.
\newblock Boosting and differential privacy.
\newblock In {\em Foundations of Computer Science (FOCS), 2010 51st Annual IEEE
  Symposium on}, pages 51--60. IEEE, 2010.

\bibitem[GLM{\etalchar{+}}10]{GLMRT10}
Anupam Gupta, Katrina Ligett, Frank McSherry, Aaron Roth, and Kunal Talwar.
\newblock Differentially private combinatorial optimization.
\newblock In {\em Proceedings of the Twenty-First Annual ACM-SIAM Symposium on
  Discrete Algorithms}, pages 1106--1125. Society for Industrial and Applied
  Mathematics, 2010.

\bibitem[GRU12]{GRU12}
Anupam Gupta, Aaron Roth, and Jonathan Ullman.
\newblock Iterative constructions and private data release.
\newblock In {\em Theory of Cryptography}, pages 339--356. Springer, 2012.

\bibitem[HHR{\etalchar{+}}14]{HHR+14}
Justin Hsu, Zhiyi Huang, Aaron Roth, Tim Roughgarden, and Zhiwei~Steven Wu.
\newblock Private matchings and allocations.
\newblock In {\em Proceedings of the 46th Annual ACM Symposium on Theory of
  Computing}, pages 21--30. ACM, 2014.

\bibitem[HLMJ09]{HLMJ09}
Michael Hay, Chao Li, Gerome Miklau, and David Jensen.
\newblock Accurate estimation of the degree distribution of private networks.
\newblock In {\em Data Mining, 2009. ICDM'09. Ninth IEEE International
  Conference on}, pages 169--178. IEEE, 2009.

\bibitem[KNRS13]{KNRS13}
Shiva~Prasad Kasiviswanathan, Kobbi Nissim, Sofya Raskhodnikova, and Adam
  Smith.
\newblock Analyzing graphs with node differential privacy.
\newblock In {\em Theory of Cryptography}, pages 457--476. Springer, 2013.

\bibitem[KRSY11]{KRSY11}
Vishesh Karwa, Sofya Raskhodnikova, Adam Smith, and Grigory Yaroslavtsev.
\newblock Private analysis of graph structure.
\newblock {\em Proceedings of the VLDB Endowment}, 4(11):1146--1157, 2011.

\bibitem[MM75]{MM75}
A~Meir and JW~Moon.
\newblock Relations between packing and covering numbers of a tree.
\newblock {\em Pacific J. Math}, 61(1):225--233, 1975.

\bibitem[NRS07]{NRS07}
Kobbi Nissim, Sofya Raskhodnikova, and Adam Smith.
\newblock Smooth sensitivity and sampling in private data analysis.
\newblock In {\em Proceedings of the thirty-ninth annual ACM symposium on
  Theory of computing}, pages 75--84. ACM, 2007.

\bibitem[War65]{Warner65}
Stanley~L Warner.
\newblock Randomized response: A survey technique for eliminating evasive
  answer bias.
\newblock {\em Journal of the American Statistical Association},
  60(309):63--69, 1965.

\end{thebibliography}

\begin{appendix}
\section{Distances in the path graph}
\label{sec:ap-dists-path}

In this section we give an explicit description of the private all-pairs distance
algorithm for the path graph $\P=(\V,\E)$ with vertex set $\V=[V]$ and edge set $\E=\{(i,i+1):i\in[V-1]\}$.
This result is a restatement 
of a result of \cite{DNPR10},
and is generalized to trees in Section \ref{ssec:DPSdists_trees}. 
This alternate argument for the special case of the path graph is included for 
illustration.

The idea behind the construction is as follows. 
We designate a small set of hubs and store more accurate distances
between consecutive hubs.
As long as we can accurately estimate the distance from 
any vertex to the nearest hub and the distance between any pair of hubs, we can 
use these distances to obtain an estimate of the distance between any pair of vertices. 
Given any pair of vertices $x,y\in \V$, in order to estimate the distance $dist(x,y)$,
we first find the hubs $h_x$ and $h_y$ nearest to $x$ and $y$. 
We estimate the distance $dist(x,y)$ by adding our estimates for the distances
$dist(x,h_x)$, $dist(h_x,h_y)$, and $dist(h_y,y)$. 

Instead of simply using a single set of hubs, we will use a hierarchical tree of hubs of different levels.
There will be many hubs of low level and only a small number of hubs of high level.
Each hub will store an estimate of the distance to the next hub at the same level and every
lower level. 
Hubs higher in the hierarchy will allow us to skip directly to distant vertices instead of accruing 
error by adding together linearly many noisy estimates. In order to estimate the distance between
a particular pair of vertices $x,y\in \V$, we will only consider a small number of hubs on each level
of the hierarchy. Since the total number of levels is not too large, this will result in a much more
accurate differentially private estimate of the distance between $x$ and $y$. 

\begin{theorem}
Let $\P=(\V,\E)$ be the path graph on $V$ vertices. For any $\epsilon>0$, 
there is algorithm $\A$ that is $\epsilon$-differentially 
private on $\P$ that on input $w:\E\to\mathbb{R}^+$
releases approximate all-pairs distances such that for each released distance,
with probability $1-\gamma$ the approximation error is 
$O(\log^{1.5} V \log(1/\gamma)) / \epsilon$ for any $\gamma\in(0,1)$. 
\end{theorem}

\begin{proof}
For fixed $k$, define  nested subsets $\V=S_0\supset S_1\supset \ldots \supset S_{k-1}$ of the vertex set $\V=[V]$ as follows. Let
$$S_i = \{x\in [V] : V^{i/k} \mid x\}.$$
That is, $S_0$ consists of all the vertices, and in general $S_i$ consists of one out of every
$V^{i/k}$ vertices on the path. Then $|S_i| = V^{(k-i)/k}$. 
Let 
$s_{i,1},s_{i,2}\ldots, s_{i,|S_i|}$
denote the elements of $S_i$ in increasing order, for each $i$. 
Using the Laplace mechanism (Lemma \ref{lemma:LapMech}),  release noisy distances between each pair of consecutive vertices
$s_{i,j}, s_{i,j+1}$ of each set, adding noise proportional to $\Lap(k/\epsilon)$. 
Note that since the vertices in each $S_i$ are in increasing order, each edge of $P$ is only considered
for a single released difference from each set.  Consequently the total sensitivity to release
all of these distances is $k$, so releasing these noisy distances is $\epsilon$-differentially private. 
Using post-processing and these special distances, we will compute approximate all-pairs distances with 
small error.

For any pair of vertices $x,y$, consider the path $\P[x,y]$ in $\P$ between $x$ and $y$,
and let $i$ be the largest index such that $S_i$ contains multiple vertices of $\P[x,y]$.
We must have that $S_i\cap \P[x,y] < 2V^{1/k}$, since otherwise
$S_{i+1}$ would contain at least two vertices on $\P[x,y]$.   
Let $x_i, y_i$ denote the first and last vertices in $S_i\cap \P[x,y]$.
For $j<i$, let $x_j$ denote the first vertex in $S_j\cap \P[x,x_i]$
and let $y_j$ denote the last vertex in $S_j\cap \P[y_i,y]$.
There are at most $1+V^{1/k}$ vertices of $S_j$ in $\P[x_j, x_{j+1}]$,
since otherwise this interval would contain another vertex of $S_{j+1}$.
Similarly, there are at most $1+V^{1/k}$ vertices of $S_j$ in $\P[y_{j+1},y_j]$.
Therefore we can express the distance from $x_j$ to $x_{j+1}$ as the
sum of at most $V^{1/k}$ distances which were estimated in $S_j$,
and similarly for the distance from $y_{j+1}$ to $y_j$.

Putting this all together, we can estimate the distance from $x=x_0$
to $y=y_0$ as the sum of at most $2(i+1)V^{1/k} \leq 2k V^{1/k}$
approximate distances which were released. But each of these distances
is released with noise distributed according to $\Lap(k/\epsilon)$.
Consequently the total error on the estimated distance from $x$ to $y$ is
the sum of at most $2kV^{1/k}$ random variables distribued according to $\Lap(k/\epsilon)$.
Taking $k=\log V$, the error is the sum of at most $4\log V$ variables distributed
according to $\Lap(\log V / \epsilon)$. By Lemma \ref{lemma:LapSumConcBd}, 
with probability at least $1-\gamma$ the sum of these 
$4\log V$ variables is bounded by $O(\log^{1.5} V \log(1/\gamma))/\epsilon$
for any $\gamma\in(0,1)$. 
Consequently this is an $\epsilon$-differentially private algorithm which releases all-pairs
distances in the path graph $P$ such that for any $\gamma\in (0,1)$, 
with probability at least $1-\gamma$ the error in each released distance is
at most $O(\log^{1.5} V \log(1/\gamma))/\epsilon$.
\end{proof}

\section{Other graph problems}
\label{sec:MST-matching-more}
In this section we consider some additional queries on graphs in the private edge weights model.


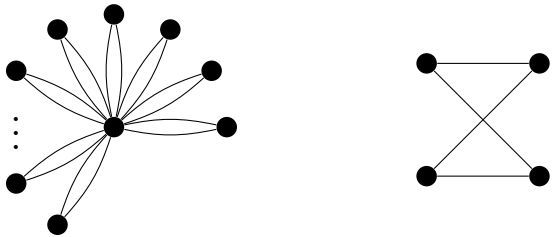
\begin{figure}[b]
\begin{center}
\raisebox{-0.5\height}{
\begin{tikzpicture}[mn/.style={circle,fill=black, scale=.75}]

\draw (0,0) node[mn] (v0) {};
\draw (0:1.5) node[mn] (v1) {};
\draw (30:1.5) node[mn] (v2) {};
\draw (60:1.5) node[mn] (v3) {};
\draw (90:1.5) node[mn] (v4) {};
\draw (120:1.5) node[mn] (v5) {};
\draw (150:1.5) node[mn] (v6) {};
\draw (210:1.5) node[mn] (v8) {};
\draw (240:1.5) node[mn] (v9) {};

\def\newvdots{\vbox{\baselineskip=7pt \lineskiplimit=0pt 
\kern6pt \hbox{.}\hbox{.}\hbox{.}}} 

\draw (180:1.3) node[mn,fill=white] (dots) {\Huge{$\newvdots$}};

\path
(v0) edge [bend right = 12] node [below] {} (v1)
(v0) edge [bend left = 12] node [above] {} (v1)
(v0) edge [bend right =12] node [below] {} (v2)
(v0) edge [bend left = 12] node [above] {} (v2)
(v0) edge [bend right = 12] node [below] {} (v3)
(v0) edge [bend left = 12] node [above] {} (v3)
(v0) edge [bend right = 12] node [below] {} (v4)
(v0) edge [bend left = 12] node [above] {} (v4)
(v0) edge [bend right = 12] node [below] {} (v5)
(v0) edge [bend left = 12] node [above] {} (v5)
(v0) edge [bend right = 12] node [below] {} (v6)
(v0) edge [bend left = 12] node [above] {} (v6)
(v0) edge [bend right = 12] node [below] {} (v8)
(v0) edge [bend left = 12] node [above] {} (v8)
(v0) edge [bend right = 12] node [below] {} (v9)
(v0) edge [bend left = 12] node [above] {} (v9);
\end{tikzpicture}}
\hspace{2cm}
\raisebox{-0.5\height}{
\begin{tikzpicture}[mn/.style={circle,fill=black, scale=.75}]

\draw (0,0) node[mn] (v0) {};
\draw (0,1.5) node[mn] (v1) {};
\draw (1.5,0) node[mn] (w0) {};
\draw (1.5,1.5) node[mn] (w1) {};

\path 
(v0) edge node [below] {} (w0)
(v0) edge node [left] {} (w1)
(v1) edge node [right] {} (w0)
(v1) edge node [above] {} (w1);

\end{tikzpicture}}

\end{center}
\caption{(Left) The graph used in the reduction of Lemma \ref{lemma:MST-reduction}. (Right) A single gadget
in the graph used in the reduction of Lemma \ref{lemma:matching-reduction}.}
\label{fig:other-problems}
\end{figure}

\subsection{Almost-minimum spanning tree}
\label{ssec:approx-MST}
We first consider the problem of releasing a low-cost spanning tree. 
The work of
\cite{NRS07} showed how to privately approximate the cost of the minimum
spanning tree in a somewhat related privacy setting. We seek to release 
an actual tree of close to minimal cost under our model.
Using techniques similar to the lower bound for shortest paths from Section \ref{ssec:DPSPlower},
we obtain a lower bound of $\Omega(V)$ for the error of releasing a low-cost spanning tree,
and show that the Laplace mechanism yields a spanning tree of cost $O(V\log V)$ more than optimal.
Note that in this section edge weights are permitted to be negative. 

\begin{theorem}
There exists a graph $\G=(\V,\E)$ such that for any spanning tree algorithm $\A$ that is $(\epsilon,\delta)$-differentially
private on $\G$, there exist edge weights $w:\E\to\binset$ such that the expected weight of the spanning tree
$\A(w)$ is at least 
$$\alpha = (V-1)\cdot\left(\frac{1-(1+e^\epsilon)\delta}{1+e^{2\epsilon}}\right)$$
longer than the weight of the minimum spanning tree. In particular, for sufficiently small $\epsilon$ and $\delta$, 
$\alpha \geq 0.49 (V-1)$. 
\label{thm:MST-lower-bound}
\end{theorem}

We first prove a lemma reducing the problem of reidentifying rows in a database to privately finding an approximate 
minimum spanning tree. 
Let $\G=(\V,\E)$ be the $(n+1)$-vertex graph with vertex set $\V=\{0,\ldots,n\}$ 
and a pair of edges
$e_i^{(0)}$ and $e_i^{(1)}$ between vertex $0$ and each vertex $i>0$, as shown in Figure \ref{fig:other-problems}.
(As in Lemma \ref{lemma:SP-reduction}, this is a multigraph, but we can transform it into 
a simple graph by adding $n$ extra vertices, changing the bound obtained by a factor of $2$.)

\begin{lemma}
Let $\G$ be the graph defined in the previous paragraph, and let $\epsilon, \delta\geq 0$. 
For any $\alpha$, let $\A$ be an algorithm that is $(\epsilon,\delta)$-differentially private on $\G$ 
that on input $w:\E\to\binset$ produces a spanning tree whose weight in
expectation is at most $\alpha$ greater than optimal. Then there exists a 
$(2\epsilon, (1+e^\epsilon)\delta)$-differentially private algorithm $\B$ which on input $x\in\binset^n$
produces $y\in\binset^n$ such that the expected Hamming distance $d_H(x,y)$ is at most $\alpha$. 
\label{lemma:MST-reduction}
\end{lemma}

\begin{proof}
The outline of the proof is the same as that of Lemma \ref{lemma:SP-reduction}.
For input $x\in\binset^n$, the corresponding edge weight function $w_x$ is given by
$w_x(e_i^{(x_i)})=0$  and $w_x(e_i^{(1-x_i)})=1$, where $x_i$ is the $i$th bit of $x$. 

We define algorithm $\B$ as follows. On input $x\in\binset^n$, apply $\A$ to $(\G,w_x)$,
and let $\T$ be the tree produced. Define $y\in\binset^n$ by setting $y_i=0$ if $e_i^{(0)}\in\T$
and $y_i=1$ otherwise. Output $y$. 

It is straightforward to verify that $\B$ is $(2\epsilon, (1+e^\epsilon)\delta)$-differentially private.
We now bound the expected Hamming distance of $x$ and $y$. 
The minimum spanning tree in $\G$ has weight $0$, so the expected weight of $\T$ is at most $\alpha$
and $\T$ must consist of at most $\alpha$ edges $e_i^{(1-x_i)}$. But $y_i\neq x_i$ only if $e_i^{(1-x_i)}\in \T$,
so in expectation $d_H(x,y) \leq w(\T) \leq \alpha$. 
\end{proof}

We now complete the proof of Theorem \ref{thm:MST-lower-bound}.
\begin{proof}[Proof of Theorem \ref{thm:MST-lower-bound}]
Assume that there is some $(\epsilon,\delta)$-differentially private algorithm which on all inputs produces a
spanning tree with expected weight less than $\alpha$ more than optimal. By Lemma \ref{lemma:MST-reduction},
there is a $(2\epsilon, (1+e^\epsilon)\delta)$-differentially private algorithm which for all
$x\in\binset^n$ produces $y\in\binset^n$ with expected Hamming distance less than $\alpha$. 
But then Lemma \ref{lemma:reidentification-Hamming}, for some $x$ the 
expected Hamming distance $\EE(d_H(x,y)) \geq 
\frac{n(1 - (1+e^{\epsilon})\delta)}{1+e^{2\epsilon}}=\alpha$, yielding a contradiction.
\end{proof}

We now show that the Laplace mechanism (Lemma \ref{lemma:LapMech}) almost matches this lower bound.

\begin{theorem}
For any $\epsilon,\gamma>0$ and $\G=(\V,\E)$, there is an algorithm $\A$ that is $\epsilon$-differentially
private on $\G$ that on input
$w:\E\to\mathbb{R}$ releases with probability $1-\gamma$ a spanning tree
of weight at most $((V-1)/\epsilon)\log(E/\gamma)$ larger than optimal.
\end{theorem}
\begin{proof}
Consider the algorithm that adds noise $X_e$ distributed according to $\Lap(1/\epsilon)$ for each edge $e\in \E$
and releases the minimum spanning tree on the resulting graph $(\G,w')$. This is $\epsilon$-differentially private,
since it is post-processing of the Laplace mechanism. We now show that the resulting error is small.
By a union bound, with probability $1-\gamma$ we have that $|X_e|\leq (1/\epsilon)\log(E/\gamma)$ for
every $e\in \E$. Consequently, conditioning on this event, if $\T$ is the spanning tree released by the algorithm
and $\T^*$ is the minimum spanning tree, then we have that
\begin{align*}
w(\T) &\leq w'(\T) \;\,+ \frac{V-1}{\epsilon} \cdot \log(E/\gamma)\\
&\leq w'(\T^*) + \frac{V-1}{\epsilon} \cdot \log(E/\gamma)\\
&\leq w(\T^*) \;+ \frac{2(V-1)}{\epsilon} \cdot \log(E/\gamma).
\end{align*}
\end{proof}

\subsection{Low-weight matching}
\label{ssec:approx-matching}

We now consider the problem of releasing a minimum weight matching in a graph in our model. 
As for the minimum spanning tree problem, a minor modification of the
lower bound for shortest paths from Section \ref{ssec:DPSPlower} yields a similar result.
For comparison, \cite{HHR+14} use similar reconstruction techniques to obtain a lower bound for 
a matching problem on bipartite graphs in a somewhat different model in which all edge weights are in $[0,1]$ 
and neighboring graphs can differ on the weights of the edges incident to a single left vertex.
We show a lower bound of $\Omega(V)$ for the error of releasing a low-weight matching tree,
and show that the matching released by the Laplace mechanism
has weight $O(V\log V)$ greater than optimal.

The theorems in this section are stated for the problem of finding a minimum weight perfect matching.
We can also obtain identical results for the problem of finding a minimum weight matching which is not required to be perfect,
and for the corresponding maximum weight matching problems. Our results apply to both bipartite matching
and general matching.
Note that in this section edge weights are permitted to be negative. 

\begin{theorem}
There exists a graph $\G=(\V,\E)$ such that for any perfect matching algorithm $\A$ which is 
$(\epsilon,\delta)$-differentially private on $\G$, there exist edge weights $w:\E\to\binset$ such that
the expected weight of the matching $\A(w)$ is at least 
$$\alpha = \left(\frac{V}{4}\right)\cdot\left(\frac{1-(1+e^\epsilon)\delta}{1+e^{2\epsilon}}\right)$$
larger than the weight of the min-cost perfect matching. In particular, for sufficiently small $\epsilon, \delta$,
$\alpha \geq 0.12 \cdot V$. 
\label{thm:matching-lower-bound}
\end{theorem}

The following lemma reduces the problem of reidentification in a database to finding a low-cost matching.
Let $\G=(\V,\E)$ be the $4n$-vertex graph with vertex set $\V=\{(b_1,b_2,c):b_1,b_2\in\binset, c\in[n]\}$
and edges from $(0,b,c)$ to $(1,b',c)$ for every $b,b'\in\binset$, $c\in[n]$. That is,
$\G$ consists of $n$ disconnected hourglass-shaped gadgets as shown in Figure \ref{fig:other-problems}.

\begin{lemma}
Let $\G=(\V,\E)$ be the graph defined in the previous paragraph. For any $\alpha$, let 
$\A$ be an algorithm that is $(\epsilon,\delta)$-differentially private on $\G$ that on input
$w:\E\to\binset$ produces a perfect matching of expected weight at most $\alpha$ greater
than optimal. Then there exists a $(2\epsilon,(1+e^\epsilon)\delta)$-differentially private algorithm $\B$
which on input $x\in\binset^n$ produces $y\in\binset^n$ with expected Hamming distance to $x$ at most $\alpha$. 
\label{lemma:matching-reduction}
\end{lemma}
\begin{proof}
For any input $x\in \binset^n$, the corresponding weight function $w_x$ assigns weight $1$ to the edge
connecting vertex $(0,1,i)$ to $(1,1-x_i,i)$ for each $i\in[n]$, and assigns weight $0$ to the other $3n$ edges.
The algorithm $\B$ is as follows. On input $x\in\binset^n$, apply $\A$ to $(\G,w_x)$, and let $\M$ be the 
matching produced. Define $y\in\binset^n$ as follows. Let $y_i=0$ if the edge from $(0,1,i)$ to $(1,0,i)$
is in the matching, and $y_i=1$ otherwise. Output $y$. 

Algorithm $\B$ is clearly $(2\epsilon, (1+e^\epsilon)\delta)$-differentially private. We will have that $y_i\neq x_i$
only if the edge from   $(0,1,i)$ to $(1,1-x_i,i)$ is in the matching, so the expected Hamming distance is at most
the expected size of the matching produced by $\A$, which is at most $\alpha$.
\end{proof}

We now conclude the proof of Theorem \ref{thm:matching-lower-bound}.

\begin{proof}[Proof of Theorem \ref{thm:matching-lower-bound}]
The result follows by combining Lemma \ref{lemma:matching-reduction} with Lemma \ref{lemma:reidentification-Hamming}.
\end{proof}

Using the Laplace mechanism (Lemma \ref{lemma:LapMech}), we obtain a nearly-matching upper bound.

\begin{theorem}
For any $\epsilon,\gamma>0$ and $\G=(\V,\E)$ containing a perfect matching, 
there is an algorithm $\A$ that is $\epsilon$-differentially private
on $\G$ that on input $w:\E\to\mathbb{R}$ releases
with probability $1-\gamma$ a perfect
matching
of weight at most $(V/\epsilon)\log(E/\gamma)$ larger than optimal.
\end{theorem}
\begin{proof}
Consider the algorithm that adds noise $X_e$ distributed according to $\Lap(1/\epsilon)$ for each edge $e\in\E$
and releases the minimum-weight perfect matching on the resulting graph $(\G,w')$. This is 
$\epsilon$-differentially private, since it is post-processing of the Laplace mechanism. We now show that the resulting error is small.
With probability $1-\gamma$ we have that $|X_e|\leq (1/\epsilon)\log(E/\gamma)$ for
every $e\in \E$. Consequently, conditioning on this event, if $\M$ is the matching released by the algorithm
and $\M^*$ is the minimum-weight matching, then we have that
\begin{align*}
w(\M)& \leq w'(\M) \:\: + \frac{V}{2\epsilon} \cdot \log(E/\gamma) \\
&\leq w'(\M^*) + \frac{V}{2\epsilon} \cdot \log(E/\gamma)\\
&\leq w(\M^*) \; + \frac{V}{\epsilon} \cdot \log(E/\gamma).
\end{align*}
\end{proof}

\end{appendix}

\end{document}